\documentclass[10pt]{article}
\usepackage{amsmath, latexsym, amssymb, amscd, amsfonts, amsthm, epsfig, mathrsfs}
\usepackage[english]{babel}

\DeclareFixedFont{\sfracFont}{U}{euf}{b}{n}{7pt}
\newtheoremstyle{mydefi}
  {15pt}
  {15pt}
  {}
  {}
  {\bfseries}
  {:}
  {.5em}
  {}
\newtheoremstyle{mytheo}
  {15pt}
  {15pt}
  {\slshape}
  {}
  {\bfseries}
  {:}
  {.5em}
  {}

\theoremstyle{mytheo}
\newtheorem{theorem}{Theorem}[section]
\newtheorem{lemma}{Lemma}
\newtheorem{proposition}{Proposition}

\theoremstyle{mydefi}
\newtheorem{definition}{Definition}
\newtheorem{example}{Example}
\newtheorem{corollary}{Corollary}
\newtheorem{assumption}{Assumption}
\newtheorem{remark}{Remark}

\begin{document}
\title{Adiabatic elimination in quantum stochastic models}
\author{Luc Bouten and Andrew Silberfarb \\ \\
\normalsize{Physical Measurement and Control 266-33, California Institute 
of Technology,}\\ \normalsize{1200 E.\ California Blvd., Pasadena, CA 
91125, USA}} \date{}
\maketitle

\begin{abstract}
We consider a physical system with a 
coupling to bosonic reservoirs via a quantum 
stochastic differential equation. 
We study the limit of this model as 
the coupling strength tends to infinity. We 
show that in this limit the solution 
to the quantum stochastic differential 
equation converges strongly to the 
solution of a limit quantum stochastic 
differential equation. In the limiting 
dynamics the excited states are removed 
and the ground states couple directly 
to the reservoirs.     
\end{abstract}

\section{Introduction}\label{sec introduction}

It is a frequent occurence in physics 
to have a system that spends a very 
limited amount of time in its excited states.   
This is, for instance, the case if the system 
is strongly coupled to a low temperature 
environment (e.g.\ the optical field). The 
strong coupling ensures that 
excitations above the ground levels 
of the system quickly dissipate 
into its environment. It is therefore 
reasonable to ask for a model in which the 
excited states are eliminated from the description.
That is, we would like to have a description 
that only involves the ground states of a 
system and its environment. The procedure for 
going from the full model to the 
reduced model is called 
\emph{adiabatic elimination}.

We study adiabatic elimination in the context 
of quantum stochastic models \cite{HuP84} which 
arise by taking a weak coupling limit of QED 
(quantum electrodynamics) models \cite{AFLu90,Gou05,DDR06}, 
and are widely applicable to systems studied in 
quantum optics.
Specifically, quantum stochastic models are the 
starting point for deriving master equations, 
filtering equations, and input-output relations.  
In the quantum optics community adiabatic elimination 
is a common technique, used, for instance, 
in atomic systems \cite{WaM94,Coh92,DeJ96,Gar84} 
and in cavity QED models \cite{GaZ00,WiM93a} as 
well as in more recent work on quantum feedback
\cite{DPTW98,DWW97,WCW97}. Rigorous results have been
demonstrated for adiabatic elimination outside of 
the quantum stochastic models we consider 
\cite{Pap76,Dav79,Gar84}. At present, however, apart 
from the work \cite{GoV07} on the elimination of 
a leaky cavity (using a Dyson series expansion to prove 
weak convergence), no rigorous results have 
been obtained on adiabatic elimination in the 
context of the quantum stochastic models introduced 
by Hudson and Parthasarathy \cite{HuP84}.

We start by considering a family, indexed by a parameter 
$k$, of quantum stochastic differential equations 
(QSDE's). The parameter $k$ can be interpreted as the 
coupling strength between the system and its 
environment. The environment is modelled by a 
collection of bosonic heat baths in the vacuum 
representation. We assume that the coefficients 
of the QSDE are all bounded and satisfy the usual conditions 
guaranteeing a unique unitary solution \cite{HuP84}. 
We state further assumptions on the coefficients and 
show that under these assumptions the solution of the QSDE 
converges strongly to the solution of a limiting QSDE as 
$k$ tends to infinity (Theorem \ref{thm main result}). 
The limiting QSDE represents the 
adiabatically eliminated time evolution of the system.  

The heart of the proof is a technique introduced by T.G.~Kurtz 
\cite{Kur73} that enables the application of 
the Trotter-Kato Theorem \cite{Tro58}. 
This allows us to prove strong convergence of the 
unitaries using convergence of generators of semigroups 
rather than convergence of a Dyson series expansion.
Convergence is first shown on the vacuum 
vector of the bosonic reservoirs. 
We then extend this result to any possible vector 
in the Hilbert space of the reservoirs by sandwiching 
the unitaries with Weyl operators and using a 
density argument. 

The remainder of this article is organized as 
follows. In Section \ref{sec main result} 
we introduce the system coupled to $n$ bosonic 
reservoirs in the vacuum representation. We 
state assumptions on the coefficients 
of the QSDE and present the main convergence 
theorem. In Section \ref{sec examples} 
we discuss four applications of the theorem 
in the context of examples from atomic physics 
and cavity QED.
Section \ref{sec proof of main result} presents 
the proof of the main convergence theorem.
In Section \ref{sec discussion} we discuss 
our results.

\section{The main result}\label{sec main result}

Let $\mathcal{H}$ be a Hilbert space and let $n$ be 
an element of $\mathbb{N}$. Let $\mathcal{F}$ 
be the \emph{symmetric Fock space} over 
$\mathbb{C}^{n}\otimes L^2(\mathbb{R}^+) 
\cong L^2(\mathbb{R}^+;\mathbb{C}^n)$, i.e.
  \begin{equation*}
  \mathcal{F} = \mathbb{C} \oplus \bigoplus_{m=1}^\infty 
  L^2(\mathbb{R}^+;\mathbb{C}^n)^{\otimes_s m}.
  \end{equation*}
Physically, the Hilbert space $\mathcal{H}\otimes\mathcal{F}$ 
describes a system $\mathcal{H}$ coupled to $n$ 
bosonic reservoirs (e.g.\ $n$ decay channels in the 
quantized electromagnetic field).  
For $f \in L^2(\mathbb{R}^+;\mathbb{C}^n)$, we 
define the \emph{exponential vector} $e(f)$ in 
$\mathcal{F}$ by 
  \begin{equation*}
  e(f) = 1\oplus \bigoplus_{m=1}^\infty \frac{f^{\otimes m}}{\sqrt{m!}}.
  \end{equation*}
Moreover, we define the \emph{coherent vector} 
$\pi(f)$ to be the exponential vector $e(f)$ 
normalized to unity, i.e.\ 
$\pi(f) = \exp(-\frac{1}{2}\|f\|^2)e(f)$.  
The \emph{vacuum vector} is defined to be
the exponential vector 
$\Phi = e(0) = 1\oplus 0\oplus 0\ldots$.
The expectation with respect to the vacuum 
vector is denoted by $\phi$, i.e.\ $\phi$ 
is a map from $\mathcal{B}(\mathcal{F})$ 
(the bounded operators on $\mathcal{F}$) to 
$\mathbb{C}$, given by $\phi(W) = \langle \Phi, W\Phi\rangle$ 
for all $W \in \mathcal{B}(\mathcal{F})$.

The interaction between the system and the 
bosonic reservoirs is modelled by a quantum 
stochastic differential equation (QSDE) 
in the sense of Hudson and Parthasarathy \cite{HuP84} 
of the form
  \begin{equation}\label{eq HuP}\begin{split}
  dU^{(k)}_t = \Bigg\{\Big(S^{(k)}_{i j}-\delta_{ij}\Big)& d \Lambda^{ij}_t + 
  L_i^{(k)}dA^{i\dagger}_{t} - L_i^{(k)\dagger} S^{(k)}_{i j}dA^j_t\ 
  + K^{(k)}dt\Bigg\}U^{(k)}_t, 
  \end{split}\end{equation}
where $U^{(k)}_0 =I$. We consistently 
use the convention that repeated 
indices that are not within parentheses 
are being summed ($i$ and $j$ run 
through $\{1,\ldots,n\}$). The Hilbert 
space adjoint is denoted by a dagger $^\dagger$. 
We have indexed the equation with a positive 
number $k$, and in the following we 
will be interested in the behaviour 
of $U_t^{(k)}$ as $k$ tends to infinity.
We assume that the following conditions on the coefficients 
of the QSDE are satisfied.

\begin{assumption}\label{as 1}
For each $k \ge 0$, the coefficients 
$K^{(k)}, S^{(k)}_{ij}$ and $L^{(k)}_i$ $(i,j \in \{1,\ldots,n\})$ 
of the quantum stochastic differential equation \eqref{eq HuP}
are bounded operators on $\mathcal{H}$. Furthermore, for each $k\ge 0$,
the coefficients satisfy the following relations
  \begin{equation*}\label{eq unit conds}
  K^{(k)}+K^{(k)\dag} = -L_i^{(k)\dag}L_i^{(k)}, 
  \qquad S^{(k)}_{i l}S^{(k)\dag}_{j l} = \delta_{ij}I,
  \qquad S^{(k)\dag}_{l i} S^{(k)}_{l j} = \delta_{ij}I.
  \end{equation*}
\end{assumption}
Hudson and Parthasarathy \cite{HuP84} show that under 
Assumption \ref{as 1}, the quantum stochastic 
differential equation \eqref{eq HuP} has a 
unique unitary solution $U_t^{(k)}$, and, the 
adjoint $U_t^{(k)\dag}$ satisfies the adjoint of 
Eq.\ \eqref{eq HuP}.

\begin{assumption}\label{as 2}
There exist bounded operators $Y,A,B,F_i,G_i$ 
and $W_{ij}$ (independent of $k$) 
on $\mathcal{H}$ such that 
  \begin{equation*}
  K^{(k)} = k^2Y+kA+B, \qquad L_i^{(k)} = kF_i + G_i, \qquad S^{(k)}_{ij} = W_{ij},
  \end{equation*}
for all $i,j \in \{1,\ldots,n\}$.
\end{assumption}
We define $P_0$ as the orthogonal projection 
onto $\mbox{Ker}(Y)$. Let $P_1 = I-P_0$ 
be its complement in $\mathcal{H}$. We use the 
following notation $\mathcal{H}_0 = P_0\mathcal{H}$ 
and $\mathcal{H}_1 = P_1\mathcal{H}$. Physically, one should 
think of $\mathcal{H}_0$ as the ground states 
and of $\mathcal{H}_1$ as the 
excited states of the system. 

\begin{assumption}\label{as 3}
There exists a bounded operator $Y_1^{-1}$ on $\mathcal{H}$ 
such that $P_1Y_1^{-1} = Y_1^{-1}P_1$ and  
  \begin{equation}\label{eq condYinv}\begin{split}
   YY_1^{-1}P_1ZP_0 =P_1ZP_0,\qquad
   P_0XP_1Y_1^{-1}Y = P_0XP_1,  
  \end{split}\end{equation}
where $Z = A, F_i^\dag W_{ij},\ (j \in \{1,\dots,n\})$ and 
$X = A,B,F_i,G_i,W_{ij}, G_i^\dag W_{ij}, F_iY_1^{-1}F_j,$ 
$F_iY_1^{-1}A, F_iY_1^{-1}F_l^\dag W_{lj}, AY_1^{-1}A, 
AY_1^{-1}F_i, AY_1^{-1}F_l^\dag W_{lj}$, $(i,j \in \{1,\dots,n\})$.
Moreover, for all $i,j \in \{1,\ldots, n\}$ the following 
products are zero 
 \begin{equation*}
  P_0YP_1= P_0AP_0 = F_iP_0 = P_0(\delta_{il} + F_i Y^{-1}_1 F^\dag_{l})W_{lj}P_1 = 0.
  \end{equation*}  
\end{assumption}
Note that the existence of $Y_1^{-1}$ satisfying 
the assumptions in Eq.\ \eqref{eq condYinv} is 
immediate if $Y$ maps $\mathcal{H}_1$ surjectively 
onto $\mathcal{H}_1$ and is therefore invertible on 
$\mathcal{H}_1$. 

\begin{definition}\label{def coefficients limitQSDE}
Suppose Assumption \ref{as 2} and \ref{as 3} hold. 
We define for all $i,j \in \{1,\ldots,n\}$ 
the following bounded operators 
on $\mathcal{H}$
  \begin{equation*}\begin{split}\label{eq limitHuP}
    &K = P_0\left(B - A  Y_1^{-1}  A \right)P_0,\qquad
    L_i = \left(G_i - F_i Y^{-1}_1 A \right) P_0,\\
    &S_{ij} = \left( \delta_{i l} + F_i Y_1^{-1} F_l^\dag \right)W_{lj} P_0.
  \end{split}\end{equation*}
\end{definition}

\begin{assumption}\label{as 4}
For all $i,j \in \{1,\ldots, n\}$ the following products are zero 
 \begin{equation*}
  P_1 L_i  = P_1 S_{ij} =0.
  \end{equation*}
\end{assumption} 

\begin{lemma}\label{lem unitarity}
Suppose that Assumption \ref{as 1}, \ref{as 2}, \ref{as 3} 
and \ref{as 4} hold. The operators in Definition 
\ref{def coefficients limitQSDE} satisfy 
\begin{equation*}
	K+K^{\dag} = -L_i^{\dag}L_i, \qquad  S_{il}S^\dag_{jl} = \delta_{ij}P_0, 
	\qquad S^\dag_{l i} S_{l j} = \delta_{i j} P_0.
\end{equation*}
\end{lemma}
\begin{proof}
By Assumptions \ref{as 1} and \ref{as 2} we have
$K^{(k)} + K^{(k)} = -L_i^{(k)\dag}L_i^{(k)},\ K^{(k)} = k^2Y+kA+B$ 
and $L^{(k)}_i = kF_i +G_i$ for all $k\ge 0$. Moreover, $F_iP_0 =0$, 
by Assumption \ref{as 3}. 
Combining these results leads to
\begin{align}
&-F_i^\dag F_i = Y + Y^\dag \nonumber
 \\& -P_1F_i^\dag G_i P_0 = P_1 (A + A^\dag) P_0 \label{eq split}
 \\&-P_0 G_i^\dag G_i P_0 = P_0 (B + B^\dag) P_0. \nonumber
 \end{align} 
We then use $YY_1^{-1}AP_0 =AP_0$ from Assumption \ref{as 3} and $L_i$ from
Definition \ref{def coefficients limitQSDE} to derive
 \begin{align*}
 	-L_i^\dag L_i &= -P_0(G_i^\dag - A^\dag  Y_1^{-1 \dag} F_i^\dag ) ( G_i - F_i
	 Y_1^{-1}  A)P_0 \\
	&= P_0 (B + B^\dag)P_0 - P_0 A^\dag  Y_1^{-1 \dag} (A + A^\dag)P_0\\
	& \quad - P_0(A+A^\dag) Y^{-1}_1  A P_0 
	  + P_0A^\dag (Y_1^{-1 \dag} + Y_1^{-1}) A P_0\\
	&= P_0 (B+B^\dag) P_0 - P_0 A  Y_1^{-1}  A P_0 \
	- P_0 A^\dag  Y_1^{-1 \dag}  A^\dag P_0\\
	&=P_0 (K + K^\dag)P_0.
\end{align*} 
By Definition \ref{def coefficients limitQSDE}
\begin{equation*}
	S_{ij} = \left( \delta_{i l} + F_i Y_1^{-1} F_l^\dag \right)W_{lj} P_0.
\end{equation*}
Combining this with $-F_i^\dag F_i = Y + Y^\dag$ from above,
\begin{align*}
S_{li}^\dag S_{lj} &= P_0 W^\dag_{m i} \left( \delta_{m l} + F_m  Y_1^{-1 \dag} F_l^\dag \right)
	\left( \delta_{l n} + F_l  Y_1^{-1} F_n^\dag \right)W_{n j} P_0\\
&= P_0 W^\dag_{l i} W_{l j} P_0 = P_0 \delta_{i j}.
\end{align*}
Then use $P_0\left( \delta_{i l} + F_i  Y_1^{-1} F_l^\dag \right)W_{lj} P_1=0$ from Assumption
\ref{as 3} and $P_1 S_{ij} P_0 = 0$ from Asumption \ref{as 4} to derive
\begin{align*}
S_{il} S^\dag_{jl} &=  P_0\left( \delta_{i n} + F_i  Y_1^{-1} F_n^\dag \right)W_{nl}
    W^\dag_{m l} \left( \delta_{m j} + F_m  Y_1^{-1 \dag} F_j^\dag \right) P_0\\
&= P_0\left( \delta_{i n} + F_i  Y_1^{-1} F_n^\dag \right)
    \left( \delta_{n j} + F_n  Y_1^{-1 \dag} F_j^\dag \right) P_0 = \delta_{i j} P_0.
\end{align*}
\end{proof}

The operators given by Definition 
\ref{def coefficients limitQSDE} are 
the coefficients of a QSDE on the Hilbert 
space $\mathcal{H}\otimes\mathcal{F}$
  \begin{equation}\label{eq HuPlimit}\begin{split}
  dU_t = \Bigg\{\Big(S_{i j}-\delta_{ij}P_0 \Big)& d \Lambda^{ij}_t + 
  L_i dA^{i\dagger}_t - L_i^{\dagger} S_{i j}dA^j_t\ 
  + K dt\Bigg\}U_t,\qquad U_0 =I. 
  \end{split}\end{equation}
Lemma \ref{lem unitarity} implies that under Assumptions 
\ref{as 1}, \ref{as 2}, \ref{as 3} and \ref{as 4}, Eq.\ \eqref{eq HuPlimit} 
has a unique unitary solution on $\mathcal{H}$ \cite{HuP84}, 
and, the adjoint $U_t^\dag$ satisfies the adjoint of Eq.\ \eqref{eq HuPlimit}.  
Moreover, $U_t$ maps $\mathcal{H}_0$ to $\mathcal{H}_0$. Note 
that $U_tP_1 =P_1$.

\begin{theorem}\label{thm main result}
Suppose Assumption \ref{as 1}, \ref{as 2}, 
\ref{as 3} and \ref{as 4} hold. Let $U_t^{(k)}$ be the unique 
unitary solution to Eq. \eqref{eq HuP}. Let 
$U_t$ be the unique unitary solution to Eq.\ \eqref{eq HuPlimit} 
where the coefficients are given by Definition 
\ref{def coefficients limitQSDE}. 
Then 
$U_t^{(k)}P_0$ converges strongly to $U_t P_0$, i.e.\
  \begin{equation*}
  \lim_{k\to \infty}U^{(k)}_t \psi = U_t \psi, \ \ \ \
  \forall \psi \in \mathcal{H}_0\otimes \mathcal{F}. 
  \end{equation*}   
\end{theorem}
We prove Theorem \ref{thm main result} in Section \ref{sec proof of main result}.

\section{Examples}\label{sec examples}

We use the following definitions 
in the first two examples below. 
Let $(|e\rangle,|g\rangle)$ be an 
orthogonal basis of $\mathbb{C}^2$. 
Define the raising and lowering operators in this
basis as
\begin{equation*}
    \sigma_+ = \begin{pmatrix} 0 & 1 \\ 0 & 0 \end{pmatrix},
     \qquad \sigma_- = \begin{pmatrix} 0 & 0 \\ 1 & 0 \end{pmatrix}.
\end{equation*}
Define the Pauli operators
\begin{equation*}
    \sigma_x = \sigma_+ + \sigma_-, \qquad
    \sigma_y = -i \sigma_+ + i \sigma_-, \qquad
    \sigma_z = \sigma_+\sigma_- - \sigma_-\sigma_+,
\end{equation*}
and define the projectors
\begin{equation*}
    P_e = \sigma_+ \sigma_-, \qquad P_g =\sigma_- \sigma_+.
\end{equation*}

\begin{example}
{\bf (A two-level atom driven by a laser)}  The Hilbert space for a two-level 
atom is $\mathcal{H} = \mathbb{C}^2$, with $|e\rangle$ 
the excited state, and $|g\rangle$ the ground state.  Define the detuning
$\Delta \in \mathbb{R}$, the decay rate $\gamma \ge 0$
and the complex amplitude $\alpha \in \mathbb{C}$.  The QSDE for
this system in the electric dipole and rotating wave
approximations is \cite{Coh92}
  \begin{equation*}\begin{split}
    dU^{(k)}_t = \Bigg\{&k\sqrt{\gamma}\sigma_{-}dA^\dag_t\ -k \sqrt{\gamma}\sigma_{+}dA_t -i k \alpha \sigma_{+}dt -i k \bar{\alpha} \sigma_-dt \\
  &-\frac{k^2\gamma}{2} \sigma_{+}\sigma_{-}dt 
  -i k^2\Delta \sigma_+ \sigma_-dt\Bigg\}U^{(k)}_t, \qquad  U^{(k)}_0 =
  I.
  \end{split}\end{equation*}
Define the operators $Y,A,B,F,G,W$ as
  \begin{equation*}\begin{split}
  &Y = (-i \Delta -\gamma/2) \sigma_+ \sigma_-,\qquad
  A = -i \alpha \sigma_+ -i \alpha \sigma_-,\qquad
  B = 0,\\
  &F = \sqrt{\gamma} \sigma_{-},\qquad
  G = 0,\qquad
  W = I.
  \end{split}\end{equation*}
This satisfies Assumptions \ref{as 1} and \ref{as 2}, and $P_0 = P_g$. We
take $Y_1^{-1} = - (i \Delta +\gamma/2)^{-1}\sigma_+\sigma_-$, and
Assumption \ref{as 3} holds by inspection. Definition \ref{def
coefficients limitQSDE} leads to the following coefficients
  \begin{equation*}\begin{split}
  K = - \frac{|\alpha|^2}{i \Delta + \gamma/2} P_g,\qquad
  L =-i \frac{\alpha\sqrt{\gamma}}{i \Delta +\gamma/2} P_g,\qquad
  S = \frac{i \Delta - \gamma/2}{i\Delta +\gamma/2} P_g.
\end{split}\end{equation*}
Note that $P_1 L = P_1 S = 0$ satisfying Assumption \ref{as 4}.
Theorem \ref{thm main result} then shows that $U_t^{(k)}P_0$
converges strongly to $U_tP_0$, given by
\begin{equation*}
    dU_t = \frac{P_g}{i \Delta + \gamma/2} \left\{
    - \gamma d \Lambda_t
    -i \alpha \sqrt{\gamma}dA^\dag_t
    + i \bar{\alpha} \sqrt{\gamma}  dA_t
    - |\alpha|^2 dt
    \right\}U_t, \qquad U_0 =
    I.
\end{equation*}

In the case that $\gamma = 0$ the two level atom decouples from the
field.  In this case we may explicitly calculate the ground state
evolution as
\begin{equation*}
    P_0 e^{-i \left(k \alpha \sigma_+ + k \bar{\alpha} \sigma_- + k^2 \Delta
    \sigma_+ \sigma_-\right) t} P_0
    = \frac{e^{-i k^2 \Delta t/2}}{\Omega}
    \left( \Omega \cos (k \Omega t) + i k \Delta \sin(k \Omega t)
    \right),
\end{equation*}
with $\Omega = \sqrt{\Delta^2 k^2 + 4 |\alpha|^2}$.  For $k \to
\infty$ this expression limits to $e^{i|\alpha|^2/\Delta}$ which is
the solution to our eliminated differential equation $dU_t = 
i\frac{|\alpha|^2}{\Delta}U_tdt,\ U_0 =I$.
\end{example}

\begin{example}
{\bf (Alkali atom)}
Now consider a system 
with Hilbert space $\mathcal{H} = \mathbb{C}^2
\otimes \mathbb{C}^2$. Physically, the system 
represents an alkali atom with no nuclear spin coupled to a driving
field on the $S_{1/2} \to P_{1/2}$ transition.  We have four
orthogonal states in this system corresponding to the 
atomic excited and ground states with angular momentum 
$m_z = \pm\frac{1}{2}$ along the $z$-axis.   
We define a detuning $\Delta \in\mathbb{R}$, 
a decay rate $\gamma \geq 0$ and a
magnetic field 
$B_i \in \mathbb{R},\ i \in {x,y,z}$.  
The system may emit into $n=3$ independent 
dipole modes, $A^i_t$, where the modes
are labelled by $i \in \{x,y,z\}$. The 
QSDE for this system in the dipole and 
rotating wave approximations is \cite{Coh92},
  \begin{equation*}\begin{split}
  dU^{(k)}_t = \Bigg\{&k\sqrt{\gamma}\sigma_- \otimes \sigma_i dA^{i\dag}_t -
  k \sqrt{\gamma}\sigma_+ \otimes \sigma_i dA^i_t   
  -\frac{3 k^2\gamma}{2} P_e \otimes I dt\\
  &-i \left(k^2\Delta P_e \otimes I  +I 
  \otimes B_i \sigma_i\right)dt\Bigg\}U^{(k)}_t, \qquad U^{(k)}_0 = I.
  \end{split}\end{equation*}
Defining the operators $Y,A,B,F_i,G_i,W_{ij}$ as
  \begin{equation*}\begin{split}
  &Y = \left(-i \Delta - \frac{3\gamma}{2}\right) P_e \otimes I,\qquad
  A = 0,\qquad
  B = -i I \otimes B_i \sigma_i \\
  &F_i = \sqrt{\gamma} \sigma_- \otimes \sigma_{i},\qquad
  G_i = 0,\qquad
  W_{ij} = \delta_{ij},
  \end{split}\end{equation*}
satisfies Assumptions \ref{as 1} and \ref{as 2}, and $P_0 = P_g\otimes I$.
We take $Y_1^{-1} = - (i \Delta + \frac{3\gamma}{2})^{-1}P_e\otimes I$, and
Assumption \ref{as 3} holds by inspection.  Define the eliminated
coefficients as
  \begin{equation*}\begin{split}
  K = -i P_g \otimes B_i \sigma_i,\qquad
  L_i =0,\qquad
  S_{i j} = P_g \otimes \left(\delta_{i j}I-
  \frac{\gamma}{i\Delta + \frac{3\gamma}{2}} \sigma_i \sigma_j\right).
\end{split}\end{equation*}
This satisfies Assumption \ref{as 4}. Theorem \ref{thm main result}
then shows that $U_t^{(k)}P_0$ converges strongly to $U_tP_0$, given
by
  \begin{equation*}
  dU_t = P_g \otimes \left\{-i B_i \sigma_idt
  - \frac{\gamma}{i\Delta +  \frac{3\gamma}{2}} \sigma_i \sigma_j 
  d \Lambda^{i j}_t \right\}U_t, \qquad U_0 = I.
  \end{equation*}
\end{example}

In the following two examples we make use of 
a truncated harmonic oscillator. We have truncated 
the oscillator to satisfy the boundedness condition 
of Assumption \ref{as 1} in the following two 
examples. Let $N$ be an element 
in $\mathbb{N}$ such that $N \ge 2$. The Hilbert space 
of the oscillator is $\mathbb{C}^{N}$. We choose 
an orthonormal basis $(|0\rangle,\ldots,|N-1\rangle)$ in 
$\mathbb{C}^N$. The \emph{annihilation operator} 
$b:\ \mathbb{C}^N \to \mathbb{C}^N$ is given by 
  \begin{equation*}
  b|n\rangle = \sqrt{n}|n-1\rangle,\qquad  \qquad n \in \{1,\ldots, N-1\},
  \end{equation*}
and $b |0\rangle =0$.
The \emph{creation operator} is defined to be the 
adjoint $b^\dag$.

\begin{example}\label{ex John and Ramon}
{\bf (Gough and Van Handel \cite{GoV07})}
Let $\mathfrak{h}$ be a Hilbert space. We 
define $\mathcal{H} = \mathfrak{h}\otimes\mathbb{C}^N$. 
The Hilbert space $\mathfrak{h}$ describes a system 
inside a cavity. We model the cavity as a truncated 
oscillator $\mathbb{C}^N$. Let $E_{ij},\ i,j \in \{0,1\}$ 
be bounded operators on $\mathfrak{h}$ such 
that $E_{ij}^\dag = E_{ji}$ and
$\|E_{11}\| < \frac{\gamma}{2}$. Consider 
the following QSDE
  \begin{equation}\label{eq our unitary}
  dU_t^{(k)} = \left\{\sqrt{\gamma}kbdA_t^\dag - 
  \sqrt{\gamma}kb^\dag dA_t - \frac{\gamma k^2}{2}b^\dag b dt 
  -iH^{(k)}dt\right\}U_t^{(k)}, 
  \qquad U_0^{(k)} = I.
  \end{equation}
Here $\gamma$ is a real parameter and $H^{(k)}$ is 
given by 
  \begin{equation*}
  H^{(k)} = k^2E_{11}b^\dag b + kE_{10} b^\dag + kE_{01} b + E_{00}.
  \end{equation*}
Define operators $Y,A,B,F,G,W$ as 
  \begin{equation*}\begin{split}
  &Y = \left(-i E_{11} - \frac{\gamma}{2}\right)b^\dag b, \qquad 
  A = -i\left(E_{10} b^\dag + E_{01} b\right), \qquad B = -i E_{00},\\
  &F = \sqrt{\gamma} b, \qquad  G= 0,\qquad W = I.
  \end{split}\end{equation*}  
This satisfies Assumptions \ref{as 1} and \ref{as 2} and 
$P_0 = I_{\mathfrak{h}} \otimes |0\rangle\langle0|$. Since 
$\|E_{11}\| < \frac{\gamma}{2}$, the inverse 
$\left(iE_{11} + \frac{\gamma}{2}\right)^{-1}$ exists.
Let $N_1^{-1}:\ \mathcal{H}_1 \to \mathcal{H}_1$ 
be the inverse of the restriction of $b^\dag b$ 
to $\mathcal{H}_1$. Taking 
$Y_1^{-1} =-\left(iE_{11} + \frac{\gamma}{2}\right)^{-1}N_1^{-1} P_1$ 
satisfies Assumption \ref{as 3}. Definition 
\ref{def coefficients limitQSDE} leads to the 
following coefficients
  \begin{equation}\label{eq our coefficients}\begin{split}
  &K = -i E_{00}P_0 - 
  E_{01}\frac{1}{iE_{11} + \frac{\gamma}{2}}E_{10}P_0, \\ 
  &L = \frac{-i\sqrt{\gamma}}{iE_{11} + \frac{\gamma}{2}}E_{10}P_0, \qquad 
  S = \frac{iE_{11} -\frac{\gamma}{2}}{iE_{11} + \frac{\gamma}{2}}P_0. 
  \end{split}\end{equation}
These coefficients satisfy Assumption \ref{as 4}. 
Theorem \ref{thm main result} then shows that 
$U_t^{(k)}P_0$ converges strongly to $U_tP_0$, 
where $U_t$ is given by 
  \begin{equation*}
  dU_t = \left\{(S-P_0)d\Lambda_t + LdA_t^\dag - L^\dag S dA_t 
  + K dt\right\}U_t,\qquad U_0 = I. 
  \end{equation*}
\end{example}

\begin{remark}
Note that we consider a truncated oscillator, where  
\cite{GoV07} treats the full 
oscillator, and that we prove our result 
strongly, whereas \cite{GoV07} proves a weak limit. The 
convergence of the Heisenberg dynamics follows 
immediately from our strong result. Apart 
from these points, Example \ref{ex John and Ramon}
reproduces the result in \cite{GoV07}. Care must 
be taken when directly comparing the limit equations, since 
the results in \cite{GoV07} are presented in the 
interaction picture with respect to the cavity. 
Under our assumptions, we define $V_t^{(k)}$ as the 
solution to 
  \begin{equation*}
  dV^{(k)}_t = \left\{\sqrt{\gamma}kbdA_t^\dag - \sqrt{\gamma}kb^\dag dA_t 
  -\frac{\gamma k^2}{2}b^\dag bdt \right\}V_t^{(k)},\qquad V_0^{(k)} = I.
  \end{equation*}   
The unitary in the interaction picture is 
then given by $\tilde{U}_t^{(k)} = 
V_t^{(k)\dag}U^{(k)}_t$, where $U_t^{(k)}$ is given by 
Eq.\ \eqref{eq our unitary}. Note that due to 
Theorem \ref{thm main result}, $V_t^{(k)}P_0$ 
converges strongly to $V_tP_0$, where $V_t$ is given by 
  \begin{equation*}
  dV_t = -2P_0d\Lambda_tV_t,\qquad V_0 = I.
  \end{equation*} 
This accounts for the sign difference between 
the coefficients in the equation for $\tilde{U}_t$ 
presented in \cite{GoV07}, and the coefficients in 
the equation for $U_t$ given by 
Eq.\ \eqref{eq our coefficients}.  
\end{remark}

\begin{example}
{\bf (Duan and Kimble \cite{DuK04})}
We again consider a system inside a 
cavity, described by the Hilbert space 
$\mathcal{H} = \mathfrak{h} \otimes \mathbb{C}^{N}$. 
The system inside the cavity is a three level 
atom, i.e.\ $\mathfrak{h} = \mathbb{C}^3$. Let 
$(|e\rangle, |+\rangle, |-\rangle)$ be an orthogonal 
basis in $\mathfrak{h}$. In this basis we define
\begin{equation*}
	\sigma_+^{(+)} = \begin{pmatrix}0&1&0\\0&0&0\\0&0&0 \end{pmatrix} \qquad
	\sigma_+^{(-)} = \begin{pmatrix}0&0&1\\0&0&0\\0&0&0 \end{pmatrix}.
\end{equation*}
Moreover define $\sigma^{(\pm)}_- = \sigma^{(\pm)\dag}_+$ 
and $P_\pm = \sigma^{(\pm)}_- \sigma^{(\pm)}_+$.
The QSDE for a lambda system with one leg 
($+ \leftrightarrow e$) resonantly coupled to the cavity, 
under the rotating wave approximation in the rotating frame, is,
  \begin{equation*}\begin{split}
  dU^{(k)}_t = \Bigg\{&
  \sqrt{\gamma} kb dA_t^\dag - \sqrt{\gamma} kb^\dag dA_t - 
  \frac{\gamma k^2}{2}b^\dag bdt\ + \\   
  &gk^2\Big(\sigma_+^{(+)}b - \sigma_-^{(+)}b^\dag\Big)dt +
  k\left(\sigma_+^{(-)}\alpha - \sigma_-^{(-)}\bar\alpha\right)dt 
  \Bigg\}U^{(k)}_t, \qquad U^{(k)}_0 = I.
  \end{split}\end{equation*}
Here $\gamma$ is a positive real parameter and $\alpha$ is 
a complex parameter. 
Note that we extend the model from \cite{DuK04} 
to allow driving on the uncoupled leg 
($- \leftrightarrow e$) of the transition. 
Define operators $Y,A,B,F,G,W$ as 
  \begin{equation*}\begin{split}
  &Y = -\frac{\gamma}{2}b^\dag b + 
  g\Big(\sigma_+^{(+)}b - \sigma_-^{(+)}b^\dag\Big), \qquad
  A = \left(\sigma_+^{(-)}\alpha - \sigma_-^{(-)}\bar\alpha\right), \qquad
  B = 0, \\
  &F = \sqrt{\gamma}b,\qquad G = 0, \qquad W = I. 
  \end{split}\end{equation*}
This satisfies Assumptions \ref{as 1} and \ref{as 2} and 
$P_0 = \Big(|+\rangle\langle+|\ +\ 
|-\rangle\langle-|\Big)\otimes |0\rangle\langle0|$.
We define the following 
subspaces of $\mathcal{H}$
  \begin{equation*}\begin{split}
  &H_n = \mbox{span}\Big\{|+\rangle\otimes|n\rangle,\,
  |-\rangle\otimes|n\rangle,\,|e\rangle\otimes|n-1\rangle\Big\},\qquad n \in 
  \{1,\ldots,N-1\}, \\
  &H_N = \mbox{span}\Big\{|e\rangle\otimes|N-1\rangle\Big\}. 
  \end{split}\end{equation*}
Note that $\mathcal{H}_1 = \bigoplus_{n=1}^N H_n$ and 
that the subspaces $H_n\ (n\in \{1,\ldots,N\})$ 
are all invariant under the action of $Y$. On the subspaces 
$H_n$, $n \in \{1,\ldots,N-1\}$, $Y$ is given by 
  \begin{equation*}
  Y|_{H_n} = \begin{pmatrix} -\frac{\gamma n}{2} &0 &-g\sqrt{n} \\
  0 &-\frac{\gamma n}{2} &0 \\
  g\sqrt{n} &0 &-\frac{\gamma(n-1)}{2}
  \end{pmatrix}, 
  \end{equation*}
with respect to the basis 
$(|+\rangle\otimes|n\rangle,
|-\rangle\otimes|n\rangle,|e\rangle\otimes|n-1\rangle)$.
Moreover, $Y|_{H_N} = -\frac{\gamma (N-1)}{2}$.
The inverse is readily computed to be 
  \begin{equation*}
  Y|_{H_n}^{-1} = -\frac{1}{d}\begin{pmatrix} \frac{\gamma(n-1)}{2} &0 &-g\sqrt{n} \\
  0 & \frac{2d}{\gamma n} &0 \\
  g\sqrt{n} &0 &\frac{\gamma n}{2}
  \end{pmatrix},\qquad n \in \{1,\ldots,N-1\},
  \end{equation*}
where $d = \frac{\gamma^2n(n-1)}{4} + g^2 n$. 
Moreover, $Y|_{H_N}^{-1} = -\frac{2}{\gamma (N-1)}$.
We now define $Y_1^{-1} = \oplus_{n=1}^{N} Y|_{H_n}^{-1}P_1$.
This satisfies Assumption \ref{as 3}.
Definition \ref{def coefficients limitQSDE} leads to 
the following coefficients
  \begin{equation*}
  K = -\frac{|\alpha|^2\gamma}{2g^2}P_-\otimes |0\rangle\langle 0|, \qquad
  L = - \frac{\gamma\alpha}{g}\sigma^{(+)}_-\sigma_+^{(-)}\otimes |0\rangle\langle 0|,\qquad
  S = P_0-2P_-\otimes|0\rangle\langle 0|.
  \end{equation*}
These operators satisfy Assumption \ref{as 4}. Theorem 
\ref{thm main result} then shows that $U_t^{(k)}P_0$ 
converges strongly to $U_tP_0$, where $U_t$ is given by 
  \begin{equation*}
  dU_t = \left\{(S-P_0)d\Lambda_t + LdA_t^\dag - L^\dag S dA_t 
  + K dt\right\}U_t,\qquad U_0 = I. 
  \end{equation*} 
Note that the ground state system is a two-level 
system on which $S$ 
acts as $\sigma_z$.   
\end{example}

\section{Proof of Theorem \ref{thm main result}}\label{sec proof of main result}

\begin{definition}\label{de semigroups}
Suppose Assumptions \ref{as 1}, \ref{as 2}, \ref{as 3}
and \ref{as 4} hold.
Let $\mathcal{B}(\mathcal{H})$ and 
$\mathcal{B}(\mathcal{H}_0)$ be the Banach 
spaces of all bounded operators on $\mathcal{H}$ and 
$\mathcal{H}_0$, respectively. We 
define for all $t \ge 0$ and $k \ge 0$ 
  \begin{equation*}\begin{split}
  &T_t^{(k)}(X)= \mbox{id}\otimes\phi\left(U^{\dag}_t X U^{(k)}_t\right),\
  \qquad X \in \mathcal{B}(\mathcal{H}), \\
  &T_t(X) = \mbox{id}\otimes\phi\left(U^{\dag}_t X U_t\right),\
  \qquad X \in \mathcal{B}(\mathcal{H}_0),
  \end{split}\end{equation*}
where $U_t^{(k)}$ and $U_t$ are given by 
Eqs.\ \eqref{eq HuP} and \eqref{eq HuPlimit}, 
respectively.  
\end{definition}

Note that $T_t^{(k)}$ is intentionally
skew  with respect to $U_t$ and $U_t^{(k)}$.

\begin{lemma}\label{lem norm continuous}
For each $k > 0$, the families of bounded linear maps 
$T_t^{(k)}\ (t\ge 0)$ and $T_t\ (t\ge 0)$ given 
by Definition \ref{de semigroups}
are norm continuous one-parameter 
contraction semigroups 
with generators
  \begin{equation}\label{eq generators}\begin{split}
  &\mathscr{L}^{(k)}(X) = K^\dag X + XK^{(k)} + L_i^\dag X L_i^{(k)},
  \qquad X \in \mathcal{B}(\mathcal{H}),\\
  &\mathscr{L}(X) = K^\dag X +XK  + L_i^\dag XL_i,\qquad X \in \mathcal{B}(\mathcal{H}_0),
  \end{split}\end{equation}
respectively. That is $T_t^{(k)} = \exp\left(t\mathscr{L}^{(k)}\right)$ 
and $T_t = \exp(t\mathscr{L})$ for all $t\ge 0$.   
\end{lemma}

\begin{proof}
We only prove the lemma for $T^{(k)}_t$. The proof
for $T_t$ can be obtained in an analogous way. 
Since the conditional expectation 
$\mbox{id}\otimes\phi$ is norm contractive 
and $U_t$ and $U_t^{(k)}$ are unitary, we 
have
  \begin{equation*}\begin{split}
  &\left\| T_t^{(k)}(X) \right\| \leq \left\|U^\dag_t X U_t^{(k)}\right\| 
  \leq \left\|U^\dag_t\right\|\left\|U^{(k)}_t\right\|\left\|X\right\| = 
  \left\|X\right\|,\\
  \end{split}\end{equation*}
for all $X \in \mathcal{B}(\mathcal{H})$. This proves that 
$T_t^{(k)}$ is a contraction for all $t \ge 0$. 
An application of the quantum It\^o rule \cite{HuP84}, 
together with the fact that vacuum expectations 
of stochastic integrals vanish, shows that 
  \begin{equation*}\begin{split}
  &dT^{(k)}_t(X) = \mbox{id}\otimes \phi \left(d\left(U_t^\dag XU_t^{(k)}\right)\right) = \\
  &\mbox{id}\otimes\phi\left(U_t^\dag\left(K^\dag X + XK^{(k)} + 
  L_i^\dag X L_i^{(k)}\right)U_t^{(k)}\right)dt = T^{(k)}_t\left(\mathscr{L}^{(k)}(X)\right)dt, 
  \end{split}\end{equation*}
for all $X \in \mathcal{B}(\mathcal{H})$. That is,
$T^{(k)}_t = \exp\left(t\mathscr{L}^{(k)}\right)$ is 
a one-parameter semigroup with generator $\mathscr{L}^{(k)}$. Furthermore,
$\mathscr{L}^{(k)}$ is bounded 
  \begin{equation*}
  \left\|\mathscr{L}^{(k)}(X)\right\| \le \left( \left\|K^\dag\right\| + 
  \left\|K^{(k)}\right\| + \left\|L_i^\dag\right\|
  \left\|L_i^{(k)}\right\| \right)\|X\|,
  \end{equation*}
which proves that $T_t^{(k)}$ is norm continuous.      
\end{proof}

The proof of Theorem \ref{thm main result} 
relies heavily on the 
Trotter-Kato theorem \cite{Tro58,Kat59} in combination with 
an argument due to Kurtz \cite{Kur73}. We have taken the formulation 
of the Trotter-Kato theorem from \cite[Thm 3.17, page 80]{Dav80}, 
see also \cite[Chapter 1, Section 6]{KuE86}. The 
formulation is more general than needed 
for the proof of Theorem \ref{thm main result}.

\begin{theorem}{\bf Trotter-Kato Theorem}\label{thm Trotter-Kato}
Let $\mathcal{B}$ be a Banach space and let 
$\mathcal{B}_0$ be a closed subspace of $\mathcal{B}$. 
For each $k \ge 0$, let $T_t^{(k)}$ be a strongly 
continuous one-parameter contraction semigroup on 
$\mathcal{B}$ with generator 
$\mathscr{L}^{(k)}$. Moreover, let $T_t$ be a 
strongly continuous one-parameter contraction 
semigroup on $\mathcal{B}_0$ with generator $\mathscr{L}$. 
Let $\mathcal{D}$ be a core for $\mathscr{L}$.
The following conditions are equivalent:
  \begin{enumerate}
  \item For all $X \in \mathcal{D}$ there exist 
  $X^{(k)} \in \mbox{Dom}\left(\mathscr{L}^{(k)}\right)$ such 
  that 
    \begin{equation*}
    \lim_{k \to \infty} X^{(k)} = X, \qquad 
    \lim_{k \to \infty} \mathscr{L}^{(k)}\left(X^{(k)}\right) = 
    \mathscr{L}(X).
    \end{equation*}
  \item For all $0 \le s < \infty$ and all $X \in \mathcal{B}_0$ 
    \begin{equation*}
    \lim_{k \to \infty} \left\{\sup_{0 \le t \le s} 
    \left\|T_t^{(k)}(X) - T_t(X)\right\|\right\} = 0.
    \end{equation*}   
  \end{enumerate}
\end{theorem}

\begin{proposition}\label{prop convergence}
Let $T_t^{(k)}$ and $T_t$ be the one-parameter 
semigroups on $\mathcal{B}(\mathcal{H})$ and 
$\mathcal{B}(\mathcal{H}_0)$ defined in 
Definition \ref{de semigroups}, respectively. 
We have
  \begin{equation*}
      \lim_{k \to \infty} \left\{\sup_{0 \le t \le s} 
    \left\|T_t^{(k)}(X) - T_t(X)\right\|\right\} = 0,
  \end{equation*}
for all $X \in \mathcal{B}(\mathcal{H}_0)$ and 
$0 \le s < \infty$.  
\end{proposition}
\begin{proof} The proof follows the line of  
the proof of \cite[Theorem 2.2]{Kur73}. Lemma 
\ref{lem norm continuous} shows that $T_t^{(k)} = \exp\left(t\mathscr{L}^{(k)}\right)$ 
and $T_t = \exp(t\mathscr{L})$ are norm continuous, and therefore 
also strongly continuous semigroups with generators 
given by Eq.\ \eqref{eq generators}. This means we satisfy the 
assumptions of the Trotter-Kato Theorem (Thm.\ \ref{thm Trotter-Kato}) 
with $\mathcal{D} = \mathcal{B}(\mathcal{H}_0)$ and 
$\mbox{Dom}\left(\mathscr{L}^{(k)}\right) = \mathcal{B}(\mathcal{H})$.

We can write $\mathscr{L}^{(k)}(X) = \mathscr{L}_0(X) + 
 k\mathscr{L}_1(X) + k^2\mathscr{L}_2(X),\  X \in \mathcal{B}(\mathcal{H})$,
where (recall Assumption \ref{as 2})
  \begin{equation*}\begin{split}
  &\mathscr{L}_0(X) = K^\dag X +XB + L_i^\dag XG_i, \ \ \ \
  \mathscr{L}_1(X) = XA + L_i^\dag XF_i, \ \ \ \
  \mathscr{L}_2(X) = XY.
  \end{split}\end{equation*}    
Let $X$ be an element in $\mathcal{B}(\mathcal{H}_0)$ and 
let $X_1$ and $X_2$ be elements in $\mathcal{B}(\mathcal{H})$.
We define $X^{(k)} = X + \frac{1}{k}X_1 + \frac{1}{k^2}X_2$.
Collecting terms with equal powers in $k$, we find
  \begin{equation*}\begin{split}
  \mathscr{L}^{(k)}\left(X^{(k)}\right) ={ } 
  &\left(\mathscr{L}_0(X) + \mathscr{L}_1(X_1) + \mathscr{L}_2(X_2)\right)\ + \\ 
  &k\left(\mathscr{L}_1(X) + \mathscr{L}_2(X_1)\right)\ + \\
  &k^2\left(\mathscr{L}_2(X)\right)\ + \\
  &\frac{1}{k}\left(\mathscr{L}_0(X_1) + \mathscr{L}_1(X_2)\right) +
  \frac{1}{k^2}\left(\mathscr{L}_0(X_2)\right).  
  \end{split}\end{equation*}
Note that $\mathscr{L}_2(X) = 0$ as $X \in \mathcal{B}(\mathcal{H}_0)$ 
and $P_0Y = 0$. Using the existence of $Y_1^{-1}$, 
we set
  \begin{equation*}\begin{split}
  &X_1 =  - \mathscr{L}_1(X)Y_1^{-1}P_1,\\ 
  &X_2 = - \left(\mathscr{L}_0(X) +\mathscr{L}_1(X_1)\right)Y_1^{-1}P_1.
  \end{split}\end{equation*}
Using the properties of $Y_1^{-1}$ in Assumption \ref{as 3}, we obtain
  \begin{equation*}\begin{split}
  \lim_{k\to \infty} \mathscr{L}^{(k)}\left(X^{(k)}\right){ } &= 
  \lim_{k\to \infty} \left(\mathscr{L}(X) +
  \frac{1}{k}\left(\mathscr{L}_0(X_1) + \mathscr{L}_1(X_2)\right) +
  \frac{1}{k^2}\mathscr{L}_0(X_2)\right) \\
  &= \mathscr{L}(X).
  \end{split}\end{equation*}
The proposition then follows from   
the Trotter-Kato Theorem.    
\end{proof}

Note that for all $v \in \mathcal{H}_0$, we can write 
$U_t v\otimes\Phi = P_0 U_tv\otimes\Phi$. This leads to
  \begin{equation*}\begin{split}
  \left\|(U^{(k)}_t -U_t)v\otimes\Phi\right\|^2{ } &=
  \left\|(U^{(k)}_t -P_0U_t)v\otimes\Phi\right\|^2 \\
  &= 
  \left\langle v,\left(2I - T_t^{(k)}(P_0) - T_t^{(k)}(P_0)^\dag\right)v\right\rangle.  
  \end{split}\end{equation*}
Here we have used that $\mbox{id}\otimes\phi$ is a positive 
map, i.e.\ it commutes with the adjoint. Using 
Proposition \ref{prop convergence} and noting that 
$\mathscr{L}(P_0)= 0$ by Lemmas \ref{lem norm continuous} 
and \ref{lem unitarity}, we see that 
Theorem \ref{thm main result} holds for 
all vectors in $\mathcal{H}_0 \otimes \mathcal{F}$ 
of the form $\psi = v\otimes \Phi$. We now need to 
extend this to all 
$\psi \in \mathcal{H}_0\otimes \mathcal{F}$. 

Let $f$ be an element in $L^2(\mathbb{R}^+;\mathbb{C}^n)$.
Denote by $f_t$ the function $f$ truncated at time $t$, i.e.\
$f_t(s) = f(s)$ if $s \le t$ and $f_t(s) = 0$ otherwise. 
Define the \emph{Weyl operator} $W(f_t)$ as the
unique solution to the following QSDE
  \begin{equation}\label{eq Weyl}
  dW(f_t) = \left\{f(t)_idA^{i\dag}_t - \overline{f(t)}_idA^i_t - 
  \frac{1}{2}\overline{f(t)}_i{f(t)_i}dt\right\}W(f_t), \qquad W(f_0) = I.
  \end{equation}
Note that $W(f_t)$ is a unitary operator from $\mathcal{F}$ 
to $\mathcal{F}$. Moreover, it is not hard to see 
that $\pi(f_t) = W(f_t)\Phi$, see e.g.\ \cite{Par92}.
Often we will identify a constant $\alpha\in\mathbb{C}^n$ 
with the constant function on $\mathbb{R}^+$ taking the 
value $\alpha$ (truncated at some large $T\ge 0$ so that 
it is an element of $L^2(\mathbb{R}^+;\mathbb{C}^n)$).

\begin{definition}\label{def Utalpha} Let $f$ be an element in 
$L^2(\mathbb{R}^+;\mathbb{C}^n)$. Suppose that 
Assumptions \ref{as 1}, \ref{as 2}, \ref{as 3} and 
\ref{as 4} hold and let $U^{(k)}_t$ and $U_t$ 
be given by Eqs.\ \eqref{eq HuP} and \eqref{eq HuPlimit}, 
respectively. Define
  \begin{equation*}\begin{split}
  &U_t^{(kf)} =  W(f_t)^\dag U_t^{(k)}W(f_t), \qquad 
  U_t^{(f)} = W(f_t)^\dag U_tW(f_t), \\
  &T^{(kf)}_t(X) = \mbox{id}\otimes\phi\left(U_t^{(f)\dag}XU_t^{(kf)}\right),
  \qquad X \in \mathcal{B}(\mathcal{H}),\\
  &T^{(kf)}_t(X) = \mbox{id}\otimes\phi\left(U_t^{(f)\dag}XU_t^{(f)}\right),
  \qquad X \in \mathcal{B}(\mathcal{H}_0).
  \end{split}\end{equation*}
\end{definition}

\begin{definition}\label{def coefficientsQSDEdisplaced}
Let $\alpha$ be an element in $\mathbb{C}^n$ 
and let $i$ be an element in $\{1,\ldots,n\}$.
Let $K^{(k)},K,L^{(k)}_i,L_i,S^{(k)}_{ij}$ and $S_{ij}$ 
be the coefficients of Eqs.\ \eqref{eq HuP} and \eqref{eq HuPlimit}.
Define operators $K^{(k\alpha)}, K^{(\alpha)}, L^{(k\alpha)}_i$ 
and $L^{(\alpha)}_i$ by 
  \begin{equation*}\begin{split}
  &K^{(\alpha)} = K + \bar{\alpha}_i (S_{ij}- P_0\delta_{i j}) \alpha_j 
  + \bar{\alpha}_i L_i - \alpha_j L^\dag_i S_{i j}, \qquad
  L_i^{(\alpha)} = L_i + \alpha_j S_{i j},\\
  &K^{(k\alpha)} = K^{(k)} + \bar{\alpha}_i (S^{(k)}_{ij}- \delta_{i j}) \alpha_j 
  + \bar{\alpha}_i L^{(k)}_i - \alpha_j L^{(k)\dag}_i S_{i j}, \ \ \
  L_i^{(k\alpha)} = L^{(k)}_i + \alpha_j S^{(k)}_{i j}.
  \end{split}\end{equation*}
\end{definition}
Note that with the coefficients given by 
Definition \ref{def coefficientsQSDEdisplaced}, applying  
the quantum It\^o rule to $U^{(k\alpha)}_t$ 
and $U^{(\alpha)}_t$, defined in Definition \ref{def Utalpha}, 
gives
  \begin{equation}\begin{split}\label{eq Hupalpha}
  &dU^{(\alpha)}_t = \Bigg\{\Big(S_{i j}-\delta_{ij}P_0\Big) d \Lambda^{ij}_t + 
  L_i^{(\alpha)}dA^{i\dagger}_{t} - L_i^{(\alpha) \dagger} S_{i j}dA^j_t\ 
  + K^{(\alpha)}dt\Bigg\}U^{(\alpha)}_t, \\
  &dU^{(k\alpha)}_t = \Bigg\{\Big(S^{(k)}_{i j}-\delta_{ij}\Big) d \Lambda^{ij}_t + 
  L_i^{(k\alpha)}dA^{i\dagger}_{t} - L_i^{(k\alpha) \dagger} S^{(k)}_{i j}dA^j_t\ 
  + K^{(k\alpha)}dt\Bigg\}U^{(k\alpha)}_t, 
  \end{split}\end{equation}
with $U_0^{(\alpha)} = U_0^{(k\alpha)} = I$.

\begin{definition}\label{def coefficients displaced}
Suppose that Assumptions \ref{as 1}, \ref{as 2}, \ref{as 3}
and \ref{as 4} hold. Let $\alpha$ be an element in 
$\mathbb{C}^n$ and let $i$ be an element in $\{1,\ldots,n\}$.
Define operators $A^{(\alpha)},B^{(\alpha)}$ 
and $G_i^{(\alpha)}$ by
  \begin{equation*}\begin{split}
  &A^{(\alpha)} = A + F_i \bar{\alpha}_i - \alpha_j F_i^\dag W_{i j},\\
  &B^{(\alpha)} =  B + \bar{\alpha}_i (W_{ij}- \delta_{i j}) \alpha_j + G_i \bar{\alpha}_i
	 - \alpha_j G^\dag_i W_{ij}, \\
  &G_i^{(\alpha)} =  G_i + \alpha_j W_{i j}.
  \end{split}\end{equation*}  
\end{definition}

\begin{lemma}\label{lem coefficients displaced}
Suppose Assumptions \ref{as 1}, \ref{as 2}, \ref{as 3}
and \ref{as 4} hold. Let $A,B,Y,F_i,G_i,W_{i j},K,L_i$ 
and $S_{i j}$ for $i,j \in \{1,\ldots,n\}$ be the 
various operators occuring in Assumption \ref{as 1}, 
\ref{as 2}, \ref{as 3} and \ref{as 4}. Let $K^{(\alpha)}$ 
and $L_i^{(\alpha)}$ for $i \in \{1,\ldots,n\}$ be 
given by Definition \ref{def coefficientsQSDEdisplaced}
and let $A^{(\alpha)}, B^{(\alpha)}$ and $G_i^{(\alpha)}$ 
for $i \in \{1,\ldots,n\}$ be given by Definition 
\ref{def coefficients displaced}. Then
\begin{subequations}
  \begin{align}
	&L_i^{(\alpha)} = (G_i^{(\alpha)} - 
	F_i Y^{-1}_1  A^{(\alpha)}) P_0,\label{eq lemmaequation1}\\
	&K^{(\alpha)} = P_0 \left(B^{(\alpha)} - 
	A^{(\alpha)}  Y_1^{-1}  A^{(\alpha)} \right)P_0,\label{eq lemmaequation2}
  \end{align}
  \end{subequations}
i.e. Definition \ref{def coefficients limitQSDE} holds with 
$A = A^{(\alpha)}$, $B = B^{(\alpha)}$, $G_i = G^{(\alpha)}_i$ 
$L_i = L_i^{(\alpha)}$ and $K = K^{(\alpha)}$.
Moreover, Assumptions \ref{as 1}, \ref{as 2}, \ref{as 3}
and \ref{as 4} hold for the altered coefficients 
with $P_0$ and $Y_1^{-1}$ unchanged. 
\end{lemma}  
\begin{proof}  
To show that Definition \ref{def coefficients limitQSDE} 
holds for the altered coefficients, substitute 
$G_i^{(\alpha)}$ and  $A^{(\alpha)}$ from
Definition \ref{def coefficients displaced}, 
and $L_i^{(\alpha)}$ from Definition 
\ref{def coefficientsQSDEdisplaced} into 
Eq.\ \eqref{eq lemmaequation1}. This gives
  \begin{equation*}
	L_i + \alpha_j S_{i j} = \left(L_i + \alpha_j W_{i j}  
	+ \alpha_j F_i Y_1^{-1} F_l^\dag W_{l j}, \right)P_0 ,
  \end{equation*}
which holds if we substitute 
$S_{i j} = \left(W_{i j} + F_i Y_1^{-1} F_l^\dag W_{l j} \right)P_0$ 
from Definition \ref{def coefficients limitQSDE}. Furthermore, 
substituting $A^{(\alpha)}$ and $B^{(\alpha)}$ from 
Definition \ref{def coefficients displaced},
and $K^{(\alpha)}$ from Definition 
\ref{def coefficientsQSDEdisplaced} into 
Eq.\ \eqref{eq lemmaequation2} gives
  \begin{equation*}\begin{split}
    &\bar{\alpha}_i S_{ij} \alpha_j + 
    \bar{\alpha}_i L_i - \alpha_j L^\dag_i S_{i j}  = 
     P_0\bar{\alpha}_i W_{ij} \alpha_jP_0 + P_0 G_i P_0\bar{\alpha}_i - 
    \alpha_j P_0 G^\dag_i W_{i j}P_0 \\
    &- P_0(F_i \bar{\alpha}_i + A ) 
     Y^{-1}_1 (A- \alpha_j F_l^\dag W_{l j})P_0 + P_0AY_1^{-1}AP_0.
  \end{split}\end{equation*}
This holds if we can show that
\begin{subequations}
  \begin{align}
	&S_{ij} = P_0 \left(W_{i j} + F_i Y_1^{-1} F^\dag_l W_{l j}\right)P_0\label{eq blah1}\\
	&L_i = P_0 \left(G_i - F_i Y_1^{-1}  A\right)P_0\label{eq blah2}\\
	&L^\dag_i S_{i j} = 
	P_0 \left(G^\dag_i W_{i j} - A Y_1^{-1} F_i^\dag W_{i j}\right)P_0.\label{eq blah3}
  \end{align}
 \end{subequations}
Equations \eqref{eq blah1} and \eqref{eq blah2} are satisfied 
by Assumption \ref{as 4} as $P_1 L_i = P_1 S_{i j} =0$.
Note that Eq.\ \eqref{eq blah3} holds if we can show
  \begin{equation*}
  L^\dag_i \left(\delta_{i l} + F_i Y_1^{-1} F^\dag_l\right)W_{lj}P_0 = 
  P_0G^\dag_l W_{lj}P_0 - P_0 A Y_1^{-1}
  F_l^\dag W_{lj}P_0.
 \end{equation*} 
Substituting $L_i$ from Definition \ref{def coefficients limitQSDE}, 
this becomes
  \begin{multline*}
 	-P_0 A^\dag Y_1^{-1 \dag} F_l^\dag W_{lj}P_0 + 
	P_0 G^\dag_i F_i  Y^{-1}_1 F^\dag_lW_{lj}P_0\\ 
	- P_0 A^\dag  Y_1^{-1 \dag} F_i^\dag F_i Y_1^{-1} F^\dag_l W_{lj}P_0 +
	 P_0 A Y_1^{-1} F_l^\dag W_{lj}P_0 =0.
  \end{multline*}
Now recall that  $P_0(A+A^\dag)P_1 = - P_0 G^\dag_i F_i P_1$, 
and $Y + Y^\dag = - F_i^\dag F_i$ 
(see Eq.\ \eqref{eq split}) by Assumptions 
\ref{as 1}, \ref{as 2} and \ref{as 3}. Moreover, 
$Y Y^{-1}_1 P_1F^\dag_l W_{lj} P_0 = P_1F^\dag_l W_{lj} P_0$ 
by Assumption \ref{as 3} which shows that 
Eq.\ \eqref{eq blah3} is satisfied.

We now show that Assumptions \ref{as 1}, \ref{as 2},
\ref{as 3} and \ref{as 4} hold for the altered 
coefficients, with $P_0$ and $Y_1^{-1}$ unchanged. 
Assumption \ref{as 1} holds for the altered 
coefficients since, by Definition \ref{def Utalpha}, 
we have $U^{(k\alpha)}_t = W(f_t)^\dag U_t^{(k)}W(f_t)$ which 
is clearly unitary. By Assumption \ref{as 2} for the original 
coefficients and Definition 
\ref{def coefficientsQSDEdisplaced} and  
\ref{def coefficients displaced}, 
we see that Assumption \ref{as 2}
holds for the altered coefficients.
Assumption \ref{as 3} on the altered coefficients is seen
to hold by direct substitution of the coefficients in
Definition \ref{def coefficientsQSDEdisplaced} and 
\ref{def coefficients displaced}, followed by application
of Assumption \ref{as 3} for the original system.
Assumption \ref{as 4} holds if 
$P_1 L_i^{(\alpha)} = P_1 L_i + \alpha_i P_1 S_{i j} = 0$, 
which follows from Assumption \ref{as 4} on the 
original system.
\end{proof}
Lemma \ref{lem coefficients displaced} shows 
that Proposition \ref{prop convergence} holds 
with $T_t^{(k\alpha)}$ and $T^{(\alpha)}_t$  
replacing $T_t^{(k)}$ and $T_t$, respectively.

\begin{corollary}\label{cor}
Suppose that Assumption \ref{as 1}, \ref{as 2}, 
\ref{as 3} and \ref{as 4} hold. Let $\alpha$ be 
an element of $\mathbb{C}^n$. We have
  \begin{equation*}
      \lim_{k \to \infty} \left\{\sup_{0 \le t \le s} 
    \left\|T_t^{(k\alpha)}(X) - T^{(\alpha)}_t(X)\right\|\right\} = 0,
  \end{equation*}
for all $X \in \mathcal{B}(\mathcal{H}_0)$ and $0\le s < \infty$.    
\end{corollary}

\noindent {\bf Proof of Theorem \ref{thm main result}.}
Let $t \ge 0$. Let $f$ be a step function in 
$L^2([0,t];\mathbb{C}^n)$, i.e.\ there 
exists an $m \in \mathbb{N}$ and $0 =t_0 < t_1<\ldots<t_m =t$
and $\alpha_1,\ldots,\alpha_m \in \mathbb{C}^n$ such 
that 
  \begin{equation*}
  s \in [t_{i-1},t_i) \Longrightarrow f(s) = \alpha_i,\qquad 
  \forall i \in \{1,\ldots,m\}.
  \end{equation*}
The cocycle property of solutions to QSDE's and 
the exponential property of the symmetric Fock 
space lead to 
  \begin{equation*}\begin{split}
  &T^{(kf)}_t(X) = T^{(k\alpha_m)}_{t_1}\ldots 
  T^{(k\alpha_1)}_{t-t_{m-1}}(X), \qquad X \in \mathcal{B}(\mathcal{H}),\\ 
  &T^{(f)}_t(X) = T^{(\alpha_m)}_{t_1}\ldots 
  T^{(\alpha_1)}_{t-t_{m-1}}(X), \qquad X \in \mathcal{B}(\mathcal{H}_0).
  \end{split}\end{equation*} 
It is easy to see that Corollary \ref{cor} also holds 
for the difference of a finite product of maps 
$T^{(k\alpha_i)}_{t_i-t_{i-1}}$ and a finite product of maps 
$T^{(\alpha_i)}_{t_i-t_{i-1}}$. This leads to
  \begin{equation*}\begin{split}
  &\lim_{k\to \infty} \left\|T_t^{(kf)}(X) - T_t^{(f)}(X)\right\| = \\
  &\lim_{k\to \infty}\left\|T^{(k\alpha_m)}_{t_1}\ldots 
  T^{(k\alpha_1)}_{t-t_{m-1}}(X)- T^{(\alpha_m)}_{t_1}\ldots 
  T^{(\alpha_1)}_{t-t_{m-1}}(X)\right\| = 0, \qquad X \in \mathcal{B}(\mathcal{H}_0).
  \end{split}\end{equation*}
This immediately yields for all step functions 
$f \in L^2([0,t];\mathbb{C}^n)$ and $v \in \mathcal{H}_0$
  \begin{equation}\label{eq resultnearly}
  \lim_{k\to \infty} U_t^{(k)}v\otimes\pi(f) = U_t v\otimes\pi(f).
  \end{equation}
Note that the step functions are dense in 
$L^2([0,t];\mathbb{C}^n)$. This means that 
Eq.\ \eqref{eq resultnearly} holds for 
all $f\in L^2([0,t];\mathbb{C}^n)$.
Now note that for all 
$f \in L^2(\mathbb{R}^+; \mathbb{C}^n)$ 
and $t \le s \le \infty$, 
we have (e.g.\ \cite{Par92})
  \begin{equation*}
  W(f_s)^\dag U^{(k)}_tW(f_s) = U^{(kf_t)}_t, \qquad W(f_s)^\dag U_tW(f_s) = U^{(f_t)}_t. 
  \end{equation*}
This means that the result in Eq.\ \eqref{eq resultnearly} is 
true for all $f\in L^2(\mathbb{R}^+;\mathbb{C}^n)$. We now 
have 
  \begin{equation*}
  \lim_{k\to\infty} U_t^{(k)}\psi = U_t\psi,
  \end{equation*}
for all $\psi$ in $\mathscr{D} = 
\mbox{span}\{v\otimes\pi(f);\ v\in \mathcal{H}_0, 
f \in L^2(\mathbb{R}^+;\mathbb{C}^n)\}$. Theorem 
\ref{thm main result} then follows from the fact 
that $\mathscr{D}$ is dense in 
$\mathcal{H}_0\otimes\mathcal{F}$ (e.g.\ \cite{Par92}).

\section{Discussion}\label{sec discussion}

In this article we have studied adiabatic 
elimination in the context of the quantum 
stochastic models introduced by Hudson 
and Parthasarathy. We have shown strong 
convergence of a quantum stochastic 
differential equation to its adiabatically 
eliminated counterpart, under four assumptions.
Physically, the first Assumption 
\ref{as 1} enforces the
unitarity of the initial QSDE model.
Assumption \ref{as 2} 
ensures an appropriate scaling in  
the coupling parameter $k$ such 
that we can distinguish excited 
and ground states in our system. 
Assumptions \ref{as 3} and \ref{as 4} 
ensure the existence of a limit
dynamics independent of $k$.  Note 
that Assumption \ref{as 4} specifically forbids any
quantum jumps which terminate in an excited state, 
the presence of which  would preclude the
construction of a valid limit dynamics.

Although a Dyson series expansion for $U^{(k)}_t$ 
(e.g.\ in terms of Maassen kernels \cite{Maa85}) 
would provide a lot of intuition  
for the results we have obtained 
(see \cite{GoV07} and \cite[Chapter 5, Section 4]{Dav80}), 
we have chosen a proof along the lines of semigroups 
and their generators. An infinitesimal 
treatment has the advantage that it 
can exploit the existence of 
results such as the quantum 
It\^o rule \cite{HuP84}, the Trotter-Kato 
Theorem \cite{Tro58,Kat59} and the technique 
due to Kurtz \cite{Kur73}.

{\bf Acknowledgement.}
We thank Mike Armen, Ramon van Handel and 
Hideo Mabuchi for stimulating discussion.
We especially thank Ramon van Handel for 
pointing out mistakes in an earlier version 
of this work.
L.B.\ is supported by the ARO under 
Grant No.\ W911NF-06-1-0378. A.S.\ 
acknowledges support by the ONR under 
Grant No.\ N00014-05-1-0420.

\bibliography{ad}
\end{document}